\documentclass[11pt]{article}
\usepackage[utf8]{inputenc}

\usepackage{textcomp}
\usepackage[russian, english]{babel}

\usepackage{amsmath, amsfonts, amssymb,amsthm}
\usepackage{verbatim}

\usepackage{float}
\usepackage{changepage}
\usepackage{array}
\usepackage{enumerate}
\usepackage{hyperref}
\usepackage{authblk}

\usepackage[nottoc,numbib]{tocbibind}

\usepackage{xcolor}

\hypersetup{colorlinks,linkcolor={blue},citecolor={red},urlcolor={blue}}

\usepackage{tikz}
\usetikzlibrary{calc, matrix, arrows,decorations.pathmorphing, positioning}
\numberwithin{equation}{section}
\theoremstyle{definition}

\newtheorem{theorem}{Theorem}[section]
\newtheorem{lemma}{Lemma}[section]

\newtheorem{definition}{Definition}[section]
\newtheorem{remark}{Remark}[section]

\newtheorem{consequence}{Consequence}[section]

\usepackage{etoolbox}
\makeatletter
\AfterEndEnvironment{definition}{\@doendpe}
\AfterEndEnvironment{theorem}{\@doendpe}
\AfterEndEnvironment{lemma}{\@doendpe}
\AfterEndEnvironment{prop}{\@doendpe}
\AfterEndEnvironment{remark}{\@doendpe}
\AfterEndEnvironment{example}{\@doendpe}
\AfterEndEnvironment{consequence}{\@doendpe}
\makeatother

\newcommand{\beq}{\begin{equation}}
\newcommand{\ee}{\end{equation}}

\newcommand{\ben}{\begin{equation*}}
\newcommand{\een}{\end{equation*}}

\newcommand{\normord}[1]{:\mathrel{\mkern2mu #1 \mkern2mu}:}

\DeclareMathOperator*{\res}{Res}
\DeclareMathOperator{\Fib}{Fib}

\DeclareMathOperator*{\Tr}{Tr}

\usepackage{paralist}

  \pltopsep=\medskipamount
\newcommand\myeq{\mathrel{\stackrel{\makebox[0pt]{\mbox{\normalfont\tiny def}}}{=}}}

\providecommand{\keywords}[1]
{
  \small	
  \textbf{\textit{Keywords---}} #1
}

\newcolumntype{M}[1]{>{\centering\arraybackslash}m{#1}}

\usepackage[explicit]{titlesec}
\titleformat{\section}
  {\normalfont\normalsize\bfseries\centering}{\thesection}{1em}{\MakeUppercase{#1}}

\usepackage{url}

\title{Semi-infinite construction of
one-dimensional lattice vertex superalgebras}

\author[]{Timur Kenzhaev \thanks{kenzhaev\_t\_d@mail.ru}}

\affil[]{Skolkovo Institute of Science and Technology, Moscow, Russia}

\date{}

\begin{document}

\pagenumbering{arabic}

\maketitle

\begin{abstract}
We construct the Feigin-Stoyanovsky (combinatorial) basis in case of one-dimensional lattice vertex superalgebras $V_{\sqrt{N}\,\mathbb{Z}}$. Our proof is based on invariance of semi-infinite monomials linear span under action of corresponding Heisenberg algebra. Semi-infinite monomials are parametrized by natural generalization of Maya diagrams~---~Fibonacci configurations on $\mathbb{Z}$, which allows us to construct a desired basis with character considerations. We also discuss some related questions such as functional realization of basic subspace's dual and representational proof of Feigin-Stoyanovsky construction in case of $V_{\sqrt{2}\,\mathbb{Z}}$.  
\\
%We give representational interpretation of these identities in terms of bosonic vertex operators $V_{\sqrt{N}}\,(z)$ and $V_{-\frac{1}{\sqrt{N}}}\,(z)$ modes.
\end{abstract}


\section{Introduction}
Feigin and Stoyanovsky showed in \cite{FS} that standard $\widehat{\mathfrak{sl}_2}'$-module  $L_{(0, 1)}$ has the basis of semi-infinite monomials
\ben
e_{i_1}\,e_{i_2}\,e_{i_3}\,\ldots,
\een
such that
\begin{enumerate}
    \item $i_1 < i_2 < i_3 < \cdots$,
    \item $i_{k + 1} - i_k \geq 2$,
    \item $i_{2k} = 0 \quad \text{for } k \gg 1$,
    \item $i_{2k + 1} = 1 \quad \text{for } k \gg 1$.
\end{enumerate}
There are several ways to prove this fact described in \cite{FS, FF_principle, Penn_Lattice, Kenzhaev_Alternative_2023}. All of them are based on considering \emph{basic subspace} (or equivalently \emph{principal subalgebra}) of $L_{(0, 1)}$. This module can be viewed as particular case of 1d lattice VOA $V_{\sqrt{2}\,\mathbb{Z}}$, which led to a natural generalization of Feigin-Stoyanovsky basis on one-dimensional even lattices in \cite{FF_principle} and structure of principal subalgebra in case of arbitrary integral lattice in \cite{Penn_Lattice}. In this article we construct Feigin-Stoyanovsky basis of arbitrary one-dimensional lattice vertex superalgebras using character considerations and invariance of semi-infinite monomials linear span under action of Heisenberg algebra. It turns out, that semi-infinite basis in $V_{\sqrt{N}\,\mathbb{Z}}$ can be conveniently parametrized by infinite Fibonacci configurations, described in \cite{Kenzhaev_Durfee_2023}:
\begin{definition}
\emph{Infinite Fibonacci configuration} of type $(\theta, l), \text{where } \:\theta, l\in\mathbb{Z} $, $\theta \geq 0,\: l > 0,\:\theta \leq l$ is a function $a\colon \mathbb{Z}\longrightarrow \{0, 1\}$, such that:
\begin{enumerate}
\item $a_i + a_{i + 1} + a_{i + 2} + \ldots + a_{i + l} \leq 1$,
\item $a_i = 0\quad \text{for } i \gg 1$,
\item $a_{-n} = 1$, if $-n \equiv \theta\; \mod l + 1,\: \text{otherwise } a_{-n} = 0\quad \text{for } n \gg 1$.
\end{enumerate}
Set of $(\theta, l)$ configurations is denoted by $\Fib^{(\theta,\, l)}_{\infty}$. Set of all infinite Fibonacci configurations of all types is denoted by $\Fib_{\infty}$.
\end{definition}
Construct by induction map $\tau\colon\Fib_{\infty} \longrightarrow \mathbb{Z}^{\mathbb{N}}$ as
\ben
\tau(a)_i = -\max \left(\left\{j\in\mathbb{Z}\hspace{2mm} |\hspace{2mm} a_j\neq 0\}\setminus\{\tau(a)_{i - 1}, \ldots, \tau(a)_{1}\right\}\right).
\een
This map assigns to any Fibonacci configuration $a$ increasing sequence of positions with ``1" of configuration $-a$.
\\\\
The goal of this article is to prove \textbf{Main Theorem}:

\begin{theorem}
\label{maintheorem}
~\
	\begin{enumerate}
\item Vertex algebra $V_{\sqrt{2N}\,\mathbb{Z}}$ has basis 
\ben
\left\{\mathbf{e}_a \hspace{2mm}|\hspace{2mm} a\in\Fib^{(N,\, 2N - 1)}_{\infty}\right\},
\een
where $\mathbf{e}_a = e_{\tau(a)_1}\,e_{\tau(a)_2}\,e_{\tau(a)_3}\ldots$ is semi-infinite monomial.
\item Vertex algebra $V_{\sqrt{2N + 1}\,\mathbb{Z}}$ has basis 
\ben
\left\{\mathbf{\Theta}_a \hspace{2mm}|\hspace{2mm} a\in\Fib^{(0,\, 2N)}_{\infty}\right\},
\een
where $\mathbf{\Theta}_a$ = $\theta_{\tau(a)_1}\,\theta_{\tau(a)_2}\,\theta_{\tau(a)_3}\ldots$ is semi-infinite monomial.
\end{enumerate}
\end{theorem}
The main idea of the proof is to show that union of the basic subspaces $W_j^{\sqrt{N}}$ (subspace which is obtained by applying modes of $V_{\sqrt{N}}\,(z)$ to vector $|j\sqrt{N}\rangle$) constitutes the whole vertex algebra:
\ben
V_{\sqrt{N}\mathbb{Z}} = \bigoplus_{j\in\mathbb{Z}}\,W^{\sqrt{N}}_j, \hspace{5mm} W_{j}^{\sqrt{N}} \subset W_{j - 1}^{\sqrt{N}}.
\een
In \cite{FS} the functional realization of restricted dual $\left(W_0^{\sqrt{2}}\right)^*$ was shown:
\ben
\left(W^{\sqrt{2}}_0\right)^* = \bigoplus_{m \geq 0} \, W_{m}^*, \hspace{5mm} W_{m}^* = \left\{g(z_1, z_2, \ldots, z_m)\left(\prod\limits_{i < j}\,(z_i - z_j)^2\right)\,dz_1\,dz_2\ldots dz_m\right\},
\een
where $g$ is symmetric polynomial. This realization may be used to identify $L_{(0, 1)}^*$ with space of “semi-infinite restricted symmetric power” of the space $\Omega^1(S^1)$ of 1-forms on circle. We obtain natural generalization of this result in \hyperref[Frealization]{Section 6}.
We also give a representational proof of Feigin-Stoyanovsky basis construction in case of $L_{(0, 1)}$ in \hyperref[Section_3]{Section 3} showing semi-infinite monomials linear span invariance under action of $\widehat{\mathfrak{sl}_2}'$, complementing the list of known proofs.  
\\\\
\textbf{Acknowledgements.} The author is grateful to Prof. Boris Feigin for the problem statement and valuable discussions.	
\section{OPE and relations}
Heisenberg Lie algebra $\mathfrak{H}$ is $\mathfrak{H} = \langle C, a_n \:|\: n\in\mathbb{Z}  \rangle_{\mathbb{C}}$ with bracket
\ben
[a_n, a_m] = \delta_{n + m, 0}\:C, \hspace{5mm} [C, \:\cdot\:] = 0. 
\een
Let $F_{\mu}$ be the Fock module over Heisenberg algebra  with highest vector $|\mu\rangle$:
\ben
c |\mu\rangle = |\mu\rangle,\hspace{5mm} h_0 |\mu\rangle = \mu|\mu\rangle.  
\een
Vertex operator superalgebra structure on $V_{\sqrt{N}\,\mathbb{Z}}\, = \bigoplus\limits_{\mu \in \sqrt{N}\,\mathbb{Z}}\, F_{\mu}$ is defined by follows:
\begin{enumerate}
    \item $(\mathbb{Z}/2\mathbb{Z})$\textbf{-gradation}:
    \ben
    V^{\Bar{0}} : = \bigoplus\limits_{\substack{\lambda\in\sqrt{N}\,\mathbb{Z} \\ \lambda^2 \equiv 0 \mod 2}}\,F_{\lambda},\hspace{1cm} V^{\Bar{1}} : = \bigoplus\limits_{\substack{\lambda\in\sqrt{N}\,\mathbb{Z} \\ \lambda^2 \equiv 1 \mod 2}}\,F_{\lambda}.
    \een
    \item (\textbf{Vacuum vector}) \, $|0\rangle$.
    \item (\textbf{Translation operator}) \, $T = \sum\limits_{n \geq 0}\, a_{- n - 1}\,a_n$. 
    \item (\textbf{Vertex operators}) Set
    \ben
\begin{aligned}
&Y(|0\rangle, z) = 0, \hspace{5mm} Y(a_{-1}\,|0\rangle) := a(z) = \sum\limits_{n\in\mathbb{Z}}\,a_n\,z^{- n - 1}.
\\
&Y(|\mu\rangle, z) = e^{\mu\, q}\,z^{\mu\,a_0}\,\exp\left( -\mu\sum\limits_{n < 0}\, \frac{a_{n}}{n}\,z^{-n}\right)\,\exp\left( -\mu\sum\limits_{n > 0}\, \frac{a_{n}}{n}\,z^{-n}\right),
\end{aligned}
    \een
\end{enumerate}
with $[a_0, q] = 1$. Due to the Strong Reconstruction Theorem \cite{Iohara} these data is enough to define a vertex superalgebra structure on $V_{\sqrt{N}\,\mathbb{Z}}$.
\\\\
OPE of two vertex operators is 
\ben
V_{\eta}(z)\,V_{\mu}(\omega) = (z - w)^{\eta\mu}\, \normord{V_{\eta}(z)\,V_{\mu}(\omega)}\,,
\een
where
\ben
\normord{V_{\eta}(z)\,V_{\mu}(\omega)} \, = e^{(\eta + \mu)\,q}\,z^{\eta\,a_0}\,w^{\mu\,a_0}\,\exp\left( -\sum\limits_{n < 0}\, \frac{a_{n}}{n}\,(\eta\,z^{-n} + \mu\,w^{-n})\right)\,\exp\left( -\sum\limits_{n > 0}\, \frac{a_{n}}{n}\,(\eta\,z^{-n} + \mu\,w^{-n})\right).
\een
There is a family of conformal vectors $\omega_{\lambda} = \frac{1}{2}\,a_{-1}^2 + \lambda\,a_{-2}$, $\lambda \in\mathbb{C}$ in $V_{\sqrt{N}\,\mathbb{Z}}$. In this work vertex algebra $V_{\sqrt{2N}\,\mathbb{Z}}$ is considered as conformal with conformal vector $\omega_0$, $V_{\sqrt{2N + 1}}\,(z)$ is considered as conformal with conformal vector $\omega_{\frac{\sqrt{2N + 1}}{2}}$. Corresponding characters (see Figures \ref{Even}, \ref{Odd}) by the definition are:
\beq
\label{Character2N}
\ch V_{\sqrt{2N}\,\mathbb{Z}} = \Tr\left(z^{a_0}\, q^{L_0^{0}}\right) = \sum\limits_{m\in\mathbb{Z}}\,\frac{z^m\,q^{Nm^2}}{(q)_{\infty}},
\ee
and
\beq
\label{Character2N+1}
\ch V_{\sqrt{2N + 1}\,\mathbb{Z}} = \Tr\left(z^{a_0}\, q^{L_0^{\frac{\sqrt{2N + 1}}{2}}}\right) = \sum\limits_{m\in\mathbb{Z}}\,\frac{z^m\,q^{(2N + 1)\frac{m(m - 1)}{2}}}{(q)_{\infty}}.
\ee

\begin{figure}[h!]
\begin{center}
\includegraphics[scale = 0.4]{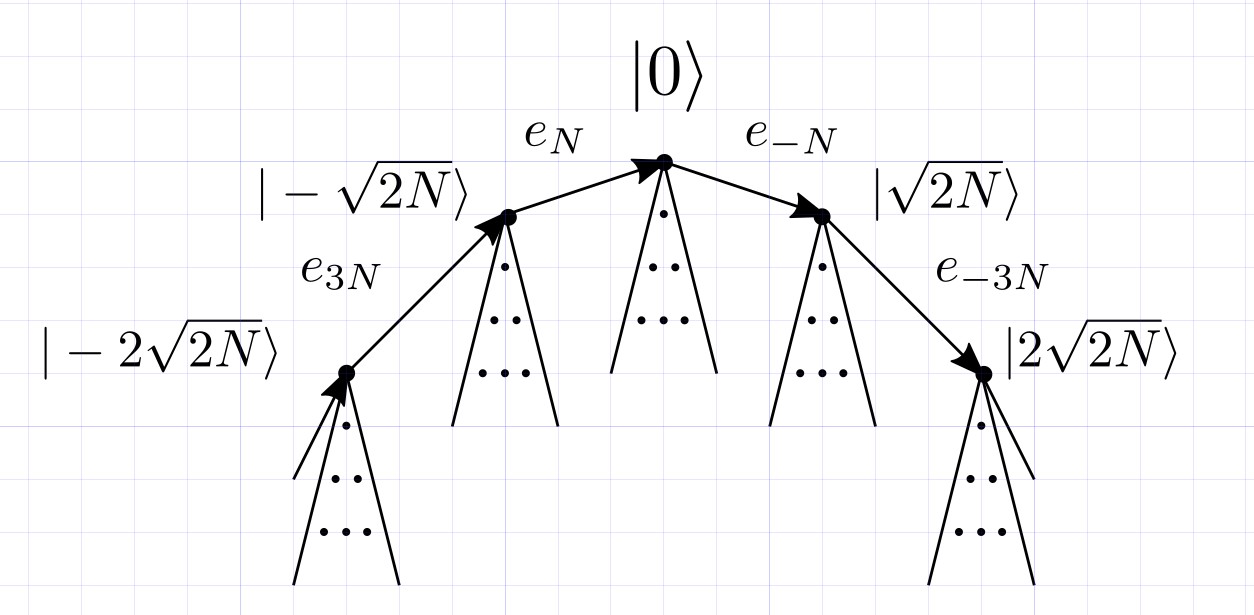}
\caption{Weight diagram for $V_{\sqrt{2N}\,\mathbb{Z}}$ \hspace{2cm} Notation from \eqref{DenoteModesEven}
\\
Conformal vector is $\omega_0 = \frac{1}{2}\,a_{-1}^2$
\\
$\ch V_{\sqrt{2N}\,\mathbb{Z}} = \sum\limits_{m\in\mathbb{Z}}\,\frac{z^m\,q^{Nm^2}}{(q)_{\infty}}$}
\label{Even}
\end{center}
\end{figure}

\begin{figure}[h!]
\begin{center}
\includegraphics[scale = 0.5]{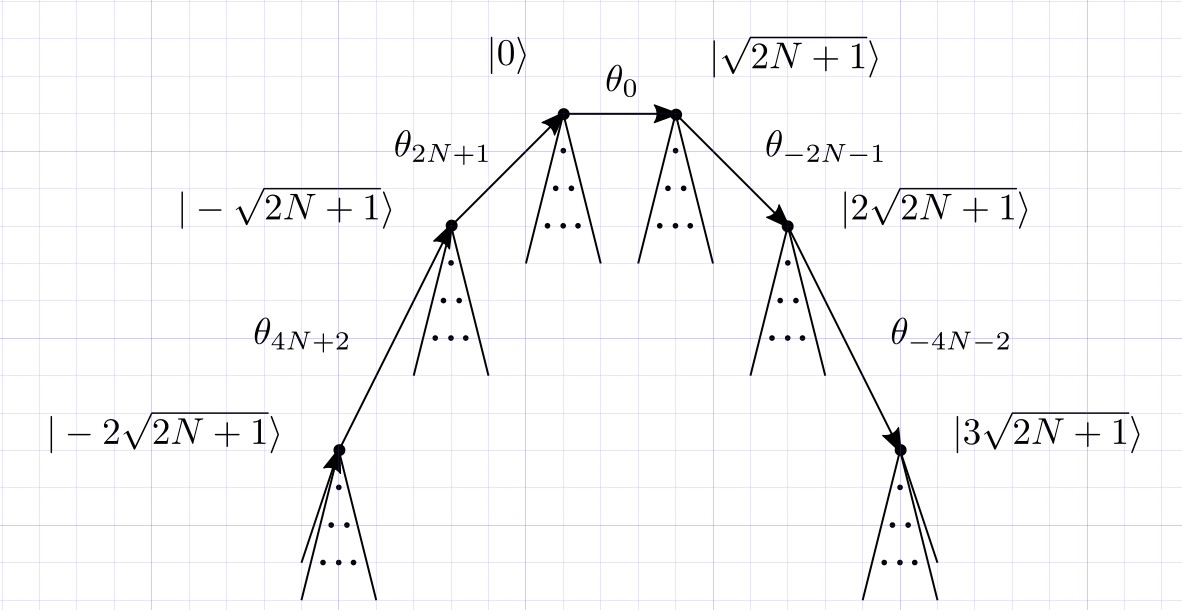}
\caption{Weight diagram for $V_{\sqrt{2N + 1}\,\mathbb{Z}}$ \hspace{2cm} Notation from \eqref{DenoteModesOdd}
\\
Conformal vector is $\omega_{\frac{\sqrt{2N + 1}}{2}} = \frac{1}{2}\,a_{-1}^2 + \frac{\sqrt{2N + 1}}{2}\,a_{-2}$
\\
$\ch V_{\sqrt{2N + 1}\,\mathbb{Z}} = \sum\limits_{m\in\mathbb{Z}}\,\frac{z^m\,q^{(2N + 1)\frac{m(m - 1)}{2}}}{(q)_{\infty}}$
}
\label{Odd}
\end{center}
\end{figure}
\noindent
As soon as $V_{\sqrt{N}}\, (z)\cdot V_{\sqrt{N}}\,(w) \simeq (z - w)^N$ there is a series of relations:
\beq
V_{\sqrt{N}}\,(z)^2 = 0, \quad V_{\sqrt{N}}\,(z)\cdot V_{\sqrt{N}}\,(z)' = 0, \quad \ldots\,, \quad V_{\sqrt{N}}\,(z)\cdot V_{\sqrt{N}}\,(z)^{(N - 1)} = 0.  
\ee
These relations are clearly not independent of each other knowing that modes of $V_{\sqrt{N}}\,(z)$ commute/anticommute. For example, in case of $N = 2$ modes of $V_{\sqrt{2}}\,(z)$ commute and
\beq
\left(V_{\sqrt{2}}\,(z)^2\right)' = 2\, V_{\sqrt{2}}\,(z)\cdot V_{\sqrt{2}}\,(z)' = 0.
\ee
Denote
\beq
\label{DenoteModesEven}
V_{\sqrt{2N}}\,(z) \equiv e(z)  = \sum\limits_{n\in\mathbb{Z}}\, e_n\,z^{ - n - N}, \hspace{5mm} V_{-\sqrt{2N}}\,(z) \equiv f(z) =  \sum\limits_{n\in\mathbb{Z}}\, f_n\,z^{- n - N}.
\ee
\beq
\label{DenoteModesOdd}
V_{\sqrt{2N + 1}}\,(z) \equiv \theta(z)  = \sum\limits_{n\in\mathbb{Z}}\, \theta_n\,z^{-n}, \hspace{5mm} V_{-\sqrt{2N + 1}}\,(z) \equiv \theta^*(z) =  \sum\limits_{n\in\mathbb{Z}}\, \theta^*_n\,z^{-n}.
\ee
In case of $V_{\sqrt{2N}\,\mathbb{Z}}$ ``defining" relations are
\beq
\label{RelationsEven}
[e_i, e_j] = 0, \quad e^2(z) = 0,\quad e\cdot e''(z) = 0, \quad \ldots\,, \quad e\cdot e^{(2N - 2)}(z) = 0. 
\ee
In case of $V_{\sqrt{2N + 1}\,\mathbb{Z}}$ ``defining" relations are
\beq
\label{RelationsOdd}
[\theta_i, \theta_j]_+ = 0, \quad \theta\cdot \theta'(z) = 0, \quad \theta\cdot \theta'''(z) = 0, \quad \ldots\,, \quad \theta\cdot \theta^{(2N - 1)}(z) = 0.
\ee
These relations are ``defining" in such sense that using them we obtain the structure of the basic subspace and consequently the semi-infinite construction of the entire vertex algebra.

\section{Semi-infinite construction of $V_{\sqrt{2}\,\mathbb{Z}}$}
\label{Section_3}
Vertex superalgebra $V_{\sqrt{2}\,\mathbb{Z}}$ is well-known to be a standard module $L_{(0, 1)}$ over $\widehat{\mathfrak{sl}}_2'$. Lie algebra $\widehat{\mathfrak{sl}_2}'$ is
\ben
\widehat{\mathfrak{sl}_2}' = \mathfrak{sl}_2\otimes\mathbb{C}[t, t^{-1}]\oplus\mathbb{C}\,K
\een
with bracket
\ben
\begin{aligned}
&[e_n, e_m] = [f_n, f_m] = 0, &[e_n, f_m] &= h_{n + m} + n\,\delta_{n, -m}\,K,
\\
&[h_n, e_m] = 2\,e_{n + m}, &[h_n, f_m] &= -2\,f_{n + m}, 
\\
&[h_n, h_m] = 2n\,\delta_{n, -m}\,K, &[K, \cdot\,\,] &= 0.
\end{aligned}
\een
Triangular decomposition is $\widehat{\mathfrak{sl}_2}' = \hat{\mathfrak{n}}_+\oplus \hat{\mathfrak{h}}\oplus \hat{\mathfrak{n}}_-$, where
\ben
\hat{\mathfrak{n}}_+ = \langle e_0\rangle + \sum\limits_{k > 0}\,t^k\,\mathfrak{sl}_2, \hspace{5mm}\hat{\mathfrak{h}} = \langle h_0, K \rangle, \hspace{5mm} \hat{\mathfrak{n}}_- = \langle f_0\rangle + \sum\limits_{k > 0}\,t^{-k}\,\mathfrak{sl}_2.
\een
\begin{definition}
$\widehat{\mathfrak{sl}_2}'$-module of the highest weight $(l, k)$ is the irreducible module $L_{(l, k)}$ with cyclic vector $v\in L_{(l, k)}$ s.t. 
\ben
\hat{\mathfrak{n}}_+\, v = 0, \quad h_0\,v = l v, \quad K\,v = k\,v.
\een
\end{definition}

\begin{definition}
\emph{Basic subspaces} of $V_{\sqrt{2}\,\mathbb{Z}}$ are 
\beq
W_j^{\sqrt{2}} \myeq \mathbb{C}[e_i \:|\: i\in\mathbb{Z}]\,|j\sqrt{2}\rangle.
\ee
\end{definition}

\begin{definition}
Consider polynomial ring $\mathbb{C}[x_i \:|\: i\in\mathbb{Z}]$. Monomial $x_{j_1}\,x_{j_2}\ldots x_{j_k} $ is called \emph{Fibonacci-1 monomial} if  $j_m - j_{m - 1} > 1$ for any $m\in\{2, 3 \ldots, k\}$. Polynomial is called Fibonacci-1 if it is a linear combination of Fibonacci monomials. Linear space of Fibonacci-1 polynomials is denoted as $\mathbb{C}^{F}_{1}[x_i]$.
There is natural bigradation on this space
\beq
\begin{aligned}
& \deg_z \left(x_{j_1}\,x_{j_2}\ldots x_{j_k}\right) = k,
\\
& \deg_q \left(x_{j_1}\,x_{j_2}\ldots x_{j_k}\right) = -\,j_1 - \,j_2 - \ldots - \,j_k.
\end{aligned}
\ee
\end{definition}
Structure of the basic subspaces is described by the following
\begin{lemma}
\label{lemmaBasicSubspaceSL}
\beq
W_{j}^{\sqrt{2}} \myeq \mathbb{C}[e_i \:|\: i\in\mathbb{Z}]\,|j\sqrt{2}\rangle = \mathbb{C}[e_i \:|\: i \leq -2j - 1]\,|j\sqrt{2}\rangle = \mathbb{C}^{F}_{1}[e_i \:|\: i \leq -2j - 1]\,|j\sqrt{2}\rangle.
\ee
\end{lemma}

\begin{proof}
As soon as $e_n |j\sqrt{2}\rangle = 0$ for $n > -2j - 1$, then
\beq
W_{j}^{\sqrt{2}} \myeq \mathbb{C}[e_i \:|\: i\in\mathbb{Z}]\,|j\sqrt{2}\rangle = \mathbb{C}[e_i \:|\: i \leq -2j - 1]\,|j\sqrt{2}\rangle.
\ee
From relations \eqref{RelationsEven} we know that $e^2(z) = 0$ which can be rewritten (up to a constant) as
\beq
\begin{aligned}
\label{RelationsE}
&e_{N}\,e_{N} \sim e_{N - 1}\,e_{N + 1} + e_{N - 2}\,e_{N + 2}
 + \ldots \hspace{1cm} N\in\mathbb{Z}, 
\\
&e_{N}\,e_{N + 1} \sim e_{N - 1}\,e_{N + 2} + e_{N - 2}\,e_{N + 3} + \ldots \hspace{1cm} N\in\mathbb{Z}. 
\end{aligned}
\ee
\begin{remark}
Right parts of relations \eqref{RelationsE} are infinite sums, nevertheless their action is correctly defined as soon as for any fixed vector in $L_{(0, 1)}$ only finite number of terms act nontrivially.
\end{remark}
Afterwards we need to prove that (considering we ignore $|j\sqrt{2}\rangle$ on the right) any monomial from ${\mathbb{C}[e_i \:|\: i \leq -2j - 1]}$ may be reduced to Fibonacci polynomial. 
For simplicity we prove this fact for $W_{0}^{\sqrt{2}}$, other cases are proved in the same way (up to a shift of indices). For any monomial $g\in\mathbb{C}[e_i \:|\: i \leq - 1]$, $g = e_{-i_1}^{j_1}\,e_{-i_2}^{j_2}\ldots\,e_{-i_n}^{j_n}$ with $j_k > 0$ and $i_1 > i_2 > \cdots > i_n$ denote
\beq
m(g) = i_1,
\ee
and suppose $m(1) = 0$.
\\
We prove by induction on $\deg_{q}(g)$ and $m(g)$ that every monomial from $\mathbb{C}[e_i \:|\: i \leq - 1]$ may be reduced to the element of $\mathbb{C}^{F}_{1}[e_i \:|\: i \leq - 1]$ with relations \eqref{RelationsE}. Monomial with $\deg_q(g) = 0$ is constant and subsequently lies in $\mathbb{C}^{F}_{1}[e_i \:|\: i \leq - 1]$. Suppose that all monomials with $\deg_q < n\in\mathbb{N}$ are reducible, then we need to prove that all monomials with $\deg_q = n$ are reducible. Monomials $g$ with $m(g) \geq n$ are all reducible as soon as they are proportional to $e_{-n}$. Suppose that all monomials $g$ with $\deg g = n$ and $m(g) > i_1$ are reducible.   
\\
Consider monomial $g = e_{-i_1}^{j_1}\,e_{-i_2}^{j_2}\ldots\,e_{-i_n}^{j_n}$ with $j_k > 0$ and $i_1 > i_2 > \cdots > i_n$, $\deg_q(g) = n$ and $m(g) = i_1$.
\beq
e_{-i_1}^{j_1}\,e_{-i_2}^{j_2}\ldots\,e_{-i_n}^{j_n} = e_{-i_1}\,e_{-i_1}^{j_1 - 1}\,e_{-i_2}^{j_2}\ldots\,e_{-i_n}^{j_n},
\ee
monomial $\tilde{g} = e_{-i_1}^{j_1 - 1}\,e_{-i_2}^{j_2}\ldots\,e_{-i_n}^{j_n}$ may be reduced to Fibonacci polynomial (as soon as $\deg_q(\tilde{g})~<~\deg_q(g)$). Then
\beq
\begin{aligned}
&e_{-i_1}\,\tilde{g} = e_{-i_1}\,e_{-i_1}\,P_{1}(e_{-i_1 + 2}, \ldots, e_{-1}) + e_{-i_1}\,e_{-i_1 + 1}\,P_{2}(e_{-i_1 + 3}, \ldots, e_{-1})
\\
&+  e_{-i_1}\,P_3(e_{-i_1 + 2}, \ldots, e_{-1}) + \sum\limits_{k}\, h_k,
\end{aligned}
\ee
where $P_1, P_2, P_3\in \mathbb{C}_{1}^{F}[e_i \:|\: i \leq - 1]$, $\sum\limits_{k}\, h_k$ is the finite sum of monomials $h_k$ with $\deg_q(h_l) = n$ and $m(g) > i_1$, i.e. they are all reducible.
\\
Summand
\beq
e_{-i_1}\,P_3(e_{-i_1 + 2}, \ldots, e_{-1}) \in \mathbb{C}_1^{F}[e_i \:|\: i \leq - 1]. 
\ee
With relations \eqref{RelationsE}
\beq
\begin{aligned}
&e_{-i_1}\,e_{-i_1}\,P_{1}(e_{-i_1 + 2}, \ldots, e_{-1}) \sim e_{-i_1 + 1}\,e_{-i_1 - 1}\,P_{1}(e_{-i_1 + 2}, \ldots, e_{-1})  
\\
& + e_{-i_1 + 2}\,e_{-i_1 - 2}\,P_{1}(e_{-i_1 + 2}, \ldots, e_{-1}) + \, \ldots \,,
\end{aligned}
\ee
\beq
\begin{aligned}
&e_{-i_1}\,e_{-i_1 + 1}\,P_{2}(e_{-i_1 + 3}, \ldots, e_{-1}) \sim e_{-i_1 - 1}\,e_{-i_1 + 2}\,P_{2}(e_{-i_1 + 3}, \ldots, e_{-1})
\\
&+ e_{-i_1 - 2}\,e_{-i_1 + 2}\,P_{2}(e_{-i_1 + 3}, \ldots, e_{-1}) + \, \ldots \,,
\end{aligned}
\ee
where all summands either vanish (being applied to $|0\rangle$) or reducible by assumption of induction (all summands have $m > i_1$ and $\deg_q = n$). 
\end{proof}

\begin{theorem}
    $W^{\sqrt{2}} = \bigcup\limits_{j\in\mathbb{Z}}\, W^{\sqrt{2}}_j$, $W^{\sqrt{2}}\subset V_{\sqrt{2}\,\mathbb{Z}} \simeq L_{(0, 1)}$ is invariant under action of $\widehat{\mathfrak{sl}_2}'$. In particular, $W^{\sqrt{2}} = V_{\sqrt{2}\,\mathbb{Z}} \simeq L_{(0, 1)}$.
\end{theorem}
\begin{proof}
\label{Theorem5.3Proof}
Denote $\widetilde{W}^{\sqrt{2}}_k = \mathbb{C}[e_i\:|\: i\geq 0]\,|k\sqrt{2}\rangle$, $k \leq 0$. For any $k\leq 0$, $\widetilde{W}^{\sqrt{2}}_{k} \subset \widetilde{W}^{\sqrt{2}}_{k - 1}$ and 
\beq
\widetilde{W}^{\sqrt{2}} = \bigcup\limits_{k \leq 0}\,\widetilde{W}_{k}^{\sqrt{2}} \subset W^{\sqrt{2}}.
\ee
If we prove invariance of $\widetilde{W}$ under $\widehat{\mathfrak{sl}_2}'$, then $\widetilde{W}^{\sqrt{2}} = L_{(0, 1)}$ and therefore $W^{\sqrt{2}} = L_{(0, 1)}$. It's clear that $\widetilde{W}^{\sqrt{2}}$ is invariant under $K$ and $e_i$ with $i\geq 0$.
\begin{enumerate}
\item Invariance under $e_{i}$ with $i < 0$.
\\
Due to $[e_i, e_j] = 0$, it's enough to prove that $e_{-l}\,|k\sqrt{2}\rangle \in \widetilde{W}$ for any $k\leq 0,\, l > 0$. 
Choose $m > 0$ s.~t. $-2k + 2m - 1 - l > 0$. 
\beq
\begin{aligned}
&e_{-l}\,|k\sqrt{2}\rangle = e_{-l}\,e_{-2k - 1}e_{-2k + 1}e_{-2k + 3}\ldots e_{-2k + 2m - 1}\,|(k - m)\sqrt{2}\rangle 
\\
&= e_{-2k - 1}e_{-2k + 1}e_{-2k + 3}\ldots e_{-l}\, e_{-2k + 2m - 1}\,|(k - m)\sqrt{2}\rangle.
\end{aligned}
\ee
Then due to $e^2(z) = 0$ relation on $L_{(0, 1)}$:
\beq
\begin{aligned}
&e_{-l}\, e_{-2k + 2m - 1}\,|(k - m)\sqrt{2}\rangle = - (e_{- l - 1}e_{-2k + 2m} + e_{- l - 2}\,e_{2k + 2m + 1} + \ldots \,)\,|(k - m)\sqrt{2}\rangle 
\\
& - (e_{-l + 1}\,e_{-2k + 2m - 2} + e_{-l + 2}\,e_{-2k + 2m - 3} + \ldots + e_{0}\,e_{-2k + 2m - 1 - l}  
\\
&+ \ldots + e_{-2k + 2m - 2}\,e_{-l + 1}) |(k - m)\sqrt{2}\rangle
\\
&- e_{-2k + 2m - 1}e_{-l} |(k - m)\sqrt{2}\rangle - (e_{-2k + 2m}e_{- l - 1} + e_{-2k + 2m + 1}e_{- l - 2})|(k - m)\sqrt{2}\rangle.
\end{aligned}
\ee
Summands
\beq
- (e_{- l - 1}e_{-2k + 2m} + e_{- l - 2}\,e_{2k + 2m + 1} + \ldots\,)\,|(k - m)\sqrt{2}\rangle
\ee
and
\beq
- (e_{-2k + 2m}e_{- l - 1} + e_{-2k + 2m + 1}e_{- l - 2})|(k - m)\sqrt{2}\rangle
\ee
vanish as soon as $e_n\,|(k - m)\sqrt{2}\rangle = 0$ for $n > -2k + 2m - 1$.
Then
\beq
\begin{aligned}
&2\,e_{-l}\, e_{-2k + 2m - 1}\,|(k - m)\sqrt{2}\rangle 
\\
& = - (e_{-l + 1}\,e_{-2k + 2m - 2} + e_{-l + 2}\,e_{-2k + 2m - 3} + \ldots + e_{0}\,e_{-2k + 2m - 1 - l}  
\\
& + \ldots + e_{-2k + 2m - 2}\,e_{-l + 1}) |(k - m)\sqrt{2}\rangle.
\end{aligned}
\ee
Thus if $\widetilde{W}^{\sqrt{2}}$ is invariant under all $e_n$ with $n \geq -l + 1$, then it's invariant under $e_{-l}$. As soon as $\widetilde{W}^{\sqrt{2}}$ is invariant under all $e_n$ with $n\geq 0$, then by induction $\widetilde{W}^{\sqrt{2}}$ is invariant under all $e_n$.
\item Invariance under $h_n$ with $n \geq 0$.
\\
As soon as $$[h_n, e_m] = 2\,e_{n + m}$$ and $h_n|k\sqrt{2}\rangle = 0$ for all $n > 0$ and $h_0|k\sqrt{2}\rangle = 2k\,|k\sqrt{2}\rangle$, $\widetilde{W}^{\sqrt{2}}$ is invariant under $h_n$ with $n \geq 0$.
\item Invariance under $f_n$ with $n \geq 0$.
\\
As soon as $[e_n, f_m] = h_{n + m} + n\,\delta_{n, -m}\,K$ and $f_n|k\sqrt{2}\rangle = 0$ for all $n > 0$ and $k \leq 0$, considering the previous point $\widetilde{W}^{\sqrt{2}}$ is invariant under $f_n$ with $n \geq 0$.
\item Invariance under $h_n$ with $n \leq 0$.
\\
As soon as $\widetilde{W}^{\sqrt{2}}$ is invariant under $e_n, \:\: n\in\mathbb{Z}$ and $f_m, \:\: m \geq 0$, considering $[e_n, f_m] = h_{n + m} + n\,\delta_{n, -m}\,K$ and commuting $e$'s with negative indices and $f$'s with positive indices we obtain invariance of $\widetilde{W}^{\sqrt{2}}$ under all $h_n$.
\item Invariance under $f_n$ with $n \leq 0$.
\\
Consequence of $[h_n, f_m] = - 2\,f_{n + m}$ and points 3 and 4. 
\end{enumerate}
\end{proof}
Then from character considerations (for details \cite{FS}, \cite{Kenzhaev_Alternative_2023},
\cite{Kenzhaev_Durfee_2023}):
$$
W_{j}^{\sqrt{2}}~\simeq~\mathbb{C}^{F}_{1}[e_i \:|\: i \leq -2j - 1].
$$
 $W_{0}^{\sqrt{2}} \simeq \mathbb{C}^{F}_{1}[e_i \:|\: i \leq -1]$ and by the same argument
$$
W_{j}^{\sqrt{2}}~\simeq~\mathbb{C}^{F}_{1}[e_i \:|\: i \leq -2j - 1],
$$
and
\begin{theorem}
$L_{(0, 1)}$ has the basis 
\beq
\left\{\mathbf{e}_a \hspace{2mm}|\hspace{2mm} a\in\Fib^{(1, 1)}_{\infty}\right\},
\ee
where $\mathbf{e}_a = e_{\tau(a)_1}\,e_{\tau(a)_2}\,e_{\tau(a)_3}\ldots$ is semi-infinite monomial.
\end{theorem}
\begin{proof}
This theorem is a particular case of Main Theorem \ref{maintheorem}.
\end{proof}

\section{Structure of the basic subspaces}
\begin{definition}
\emph{Basic subspaces} of $V_{\sqrt{2N}\,\mathbb{Z}}$ are 
\beq
W_{j}^{\sqrt{2N}} \myeq \mathbb{C}[e_i \:|\: i\in\mathbb{Z}]\,|j\sqrt{2N}\rangle,\quad j\in\mathbb{Z}.
\ee
\emph{Basic subspaces} of $V_{\sqrt{2N + 1}\,\mathbb{Z}}$ are
\beq
W_{j}^{\sqrt{2N + 1}} \myeq \mathbb{C}\{\theta_i \:|\: i\in\mathbb{Z}\}\,|j\sqrt{2N + 1}\rangle,\quad j\in\mathbb{Z}.
\ee
\end{definition}

\begin{definition}
Consider polynomial ring $\mathbb{C}[x_i \:|\: i\in\mathbb{Z}]$. Monomial $x_{j_1}\,x_{j_2}\ldots x_{j_k} $ with ${j_1 < j_2 < \cdots < j_k}$ is called \emph{Fibonacci-l monomial} if $j_m - j_{m - 1} > l$ for any ${m\in\{2, 3 \ldots, k\}}$. 
Polynomial is called Fibonacci-l if it is a linear combination of Fibonacci monomials. Linear space of Fibonacci-l polynomials is denoted as $\mathbb{C}^{F}_{l}[x_i]$. Linear space of anti-symmetric Fibonacci-l polynomials is denoted as $\mathbb{C}^{F}_{l}\{x_i\}$. There is natural bigradation on these spaces
\beq
\begin{aligned}
& \deg_z \left(x_{j_1}\,x_{j_2}\ldots x_{j_k}\right) = k,
\\
& \deg_q \left(x_{j_1}\,x_{j_2}\ldots x_{j_k}\right) = - j_1 - j_2 - \ldots - j_k.
\end{aligned}
\ee
\end{definition}

\begin{lemma}
\label{LemmaBasicSubspace2}
\beq
\begin{aligned}
&W_{j}^{\sqrt{2N}} \myeq \mathbb{C}[e_i \:|\: i\in\mathbb{Z}]\,|j\sqrt{2N}\rangle = \mathbb{C}[e_i \:|\: i \leq -2Nj - N]\,|j\sqrt{2N}\rangle 
\\
& = \mathbb{C}^{F}_{2N - 1}[e_i \:|\: i \leq -2Nj - N]\,|j\sqrt{2N}\rangle.
\end{aligned}
\ee
\beq
\begin{aligned}
&W_{j}^{\sqrt{2N + 1}} \myeq \mathbb{C}\{\theta_i \:|\: i\in\mathbb{Z}\}\,|j\sqrt{2N + 1}\rangle = \mathbb{C}\{\theta_i \:|\: i \leq -j(2N + 1)\}\,|j\sqrt{2N + 1}\rangle 
\\
& = \mathbb{C}^{F}_{2N}\{\theta_i \:|\: i \leq -j(2N + 1)\}\,|j\sqrt{2N + 1}\rangle.
\end{aligned}
\ee
\end{lemma}

\begin{proof}
1. Even case.
\\\\
According to relations \eqref{RelationsEven}, in case of $V_{\sqrt{2N}\mathbb{Z}}$ vertex algebra there is a series of relations on $e(z)$:
\beq
\label{RelationsEven2N}
e^2(z) = 0, \quad e\cdot e''(z) = 0, \quad \ldots\,, \quad e\cdot e^{(2N - 2)}(z) = 0.  
\ee
Denote $x^{(k)} = x(x + 1)(x + 2)\ldots(x + k - 1)$.
These relations can be written as series of relations
\beq
\label{EvenNN}
\begin{aligned}
&\frac{1}{2}\,e_n\,e_n + \sum\limits_{k = 1}^{\infty}\,e_{n - k}\,e_{n + k}  = 0,
\\
&(N + n)^{(2)}\,e_n\,e_n + \sum\limits_{k = 1}^{\infty} \left((N + n + k)^{(2)} + (N + n - k)^{(2)}\right)\,e_{n - k}\,e_{n + k} = 0,
\\
&(n + N)^{(4)}\,e_n\,e_n + \sum\limits_{k = 1}^{\infty} \left((N + n + k)^{(4)} + (N + n - k)^{(4)}\right)\,e_{n - k}\,e_{n + k} = 0,
\\
&\hspace{5cm}\ldots
\\
&(N + n)^{(2N - 2)}\,e_n\,e_n + \sum\limits_{k = 1}^{\infty} \left((N + n + k)^{(2N - 2)} + (2N - k)^{(2N - 2)}\right)\,e_{n - k}\,e_{n + k} = 0,
\end{aligned}
\ee
for $n\in\mathbb{Z}$ and
\beq
\label{EvenNN+1}
\begin{aligned}
&\sum\limits_{k = 0}^{\infty} \,e_{n - k}\,e_{n + 1 + k} = 0,
\\
&\sum\limits_{k = 0}^{\infty} \, \left((N + n - k)^{(2)} + (N + n + k + 1)^{(2)}\right)\,e_{n - k}\,e_{n + 1 + k} = 0,
\\
&\sum\limits_{k = 0}^{\infty} \, \left((N + n - k)^{(4)} + (N + n + k + 1)^{(4)}\right)\,e_{n - k}\,e_{n + 1 + k} = 0,
\\
&\hspace{5cm}\ldots
\\
&\sum\limits_{k = 0}^{\infty} \, \left((N + n - k)^{(2N - 2)} + (N + n + k + 1)^{(2N - 2)}\right)\,e_{n - k}\,e_{n + 1 + k} = 0,
\end{aligned}
\ee
for $n\in\mathbb{Z}$.
\\\\
In matrix form expressions \eqref{EvenNN} are
\begin{tiny}
\beq
\renewcommand{\arraystretch}{2.8}
\begin{pmatrix}
2^{-1} & 1 & 1 & \ldots & \\
2^{-1} \left((N + n)^{(2)} + (N + n)^{(2)}\right) & (N + n + 1)^{(2)} + (N + n - 1)^{(2)} & (N + n + 2)^{(2)} + (N + n - 2)^{(2)} & \ldots & \\
2^{-1}\left((N + n)^{(4)} + (N + n)^{(4)}\right) & (N + n + 1)^{(4)} + (N + n - 1)^{(4)} & (N + n + 2)^{(4)} + (N + n - 2)^{(4)} & \ldots & \\
\vdots & \vdots & \vdots & \vdots & \\
2^{-1} \left((N + n)^{(2N - 2)} + (N + n)^{(2N - 2)}\right) & (N + n + 1)^{(2N - 2)} + (N + n - 1)^{(2N - 2)} & (N + n + 2)^{(2N - 2)} + (N + n - 2)^{(2N - 2)} & \ldots &
\end{pmatrix}
\begin{pmatrix}
e_n\,e_n\\
e_{n - 1}\,e_{n + 1}\\
e_{n - 2}\,e_{n + 2}\\
\vdots \\
\vdots
\end{pmatrix}
= 0.
\ee
\end{tiny}
Leading principal $N\times N$ minor of this matrix is nondegenerate. Indeed,\\ $(N + n - k)^{(2l)} + (N + n + k)^{(2l)} = 2k^{2l} + O(k^{2l - 2})$ is even polynomial in $k$. Then by the elementary transformations determinant of this minor may be reduced to even-Vandermonde determinant
\beq
\label{EvenVandermonde}
\begin{vmatrix}
1 & 1 & \ldots & 1\\
x_0^2 & x_1^2 & \ldots & x_{N - 1}^2\\
x_0^4 & x_1^4 & \ldots & x_{N - 1}^4\\
\vdots & \vdots & \vdots & \vdots\\
x_0^{2N - 2} & x_1^{2N - 2} & \ldots & x_{N - 1}^{2N - 2}
\end{vmatrix}
\ee
at $x_i = i$.
\\\\
In matrix form expressions \eqref{EvenNN+1} are
\begin{tiny}
\beq
\renewcommand{\arraystretch}{2.5}
\begin{pmatrix}
1 & 1 & 1 & \ldots \\
(N + n)^{(2)} + (N + n + 1)^{(2)} & (N + n - 1)^{(2)} + (N + n + 2)^{(2)} & (N + n - 2)^{(2)} + (N + n + 3)^{(2)} & \ldots & \\
(N + n)^{(4)} + (N + n + 1)^{(4)} & (N + n - 1)^{(4)} + (N + n + 2)^{(4)} & (N + n - 2)^{(4)} + (N + n + 3)^{(4)} & \ldots & \\
\vdots & \vdots & \vdots & \vdots &\\
(N + n)^{(2N - 2)} + (N + n + 1)^{(2N - 2)} & (N + n - 1)^{(2N - 2)} + (N + n + 2)^{(2N - 2)} & (N + n - 2)^{(2N - 2)} + (N + n + 3)^{(2N - 2)} & \ldots &
\end{pmatrix}
\begin{pmatrix}
e_n\,e_{n + 1}\\
e_{n - 1}\,e_{n + 2}\\
e_{n - 2}\,e_{n + 3}\\
\vdots \\
\vdots 
\end{pmatrix}
= 0.
\ee
\end{tiny}
Leading principal $N\times N$ minor of this matrix is nondegenerate. \\Indeed, $(N + n - k)^{(2l)} + (N + n + k + 1)^{(2l)} = 2\left(k + \frac{1}{2}\right)^{2l} + O((k + \frac{1}{2})^{2l - 2})$ is even polynomial in $k + \frac{1}{2}$. Then by the elementary transformations determinant of this minor may be reduced to even-Vandermonde determinant \eqref{EvenVandermonde} at $x_i = i + \frac{1}{2}$.
\\
Nondegeneracy of leading principal $N\times N$ minor means that any monomial of the form $e_{n}\,e_{n + k}$ with $|k| \leq 2N - 1$ may be reduced to infinite linear combination of $e_i\,e_j$ with $|i - j| > 2N - 1$ with relations \eqref{RelationsEven2N}, which action is correctly defined. Applying argument from the \hyperref[Lemma5.1Proof]{proof} of Lemma \ref{lemmaBasicSubspaceSL} we get
\beq
W_{j}^{\sqrt{2N}}  = \mathbb{C}^{F}_{2N - 1}[e_i \:|\: i \leq -2Nj - N]\,|j\sqrt{2N}\rangle.
\ee
1. Odd case.
\\\\
According to  Theorem \ref{RelationsOdd}, in case of $V_{\sqrt{2N + 1}\,\mathbb{Z}}$ vertex algebra there is a series of relations on $\theta(z)$:
\beq
\label{RelationsOdd2N}
\theta\cdot \theta'(z) = 0, \quad \theta\cdot \theta'''(z) = 0, \quad \ldots\,, \quad \theta\cdot \theta^{(2N - 1)}(z) = 0.  
\ee
These relations can be written as series of relations
\beq
\label{OddNN}
\begin{aligned}
&n^{(1)}\theta_n\,\theta_n + \sum\limits_{k = 1}^{\infty}\,\left((n + k)^{(1)} - (n - k)^{(1)}\right)\theta_{n - k}\,\theta_{n + k}  = 0,
\\
&n^{(3)}\,\theta_n\,\theta_n + \sum\limits_{k = 1}^{\infty} \left((n + k)^{(3)} - (n - k)^{(3)}\right)\,\theta_{n - k}\,\theta_{n + k} = 0,
\\
&n^{(5)}\,\theta_n\,\theta_n + \sum\limits_{k = 1}^{\infty} \left((n + k)^{(5)} - (n - k)^{(5)}\right)\,\theta_{n - k}\,\theta_{n + k} = 0,
\\
&\hspace{4cm}\ldots
\\
&n^{(2N - 1)}\,\theta_n\,\theta_n + \sum\limits_{k = 1}^{\infty} \left((n + k)^{(2N - 1)} - (n - k)^{(2N - 1)}\right)\,\theta_{n - k}\,\theta_{n + k} = 0,
\end{aligned}
\ee
for $n\in\mathbb{Z}$ and
\beq
\label{OddNN+1}
\begin{aligned}
&\sum\limits_{k = 0}^{\infty} \,\left((n + 1 + k)^{(1)} - (n - k)^{(1)}\right) \theta_{n - k}\,\theta_{n + 1 + k} = 0,
\\
&\sum\limits_{k = 0}^{\infty} \,\left((n + 1 + k)^{(3)} - (n - k)^{(3)}\right) \theta_{n - k}\,\theta_{n + 1 + k} = 0,
\\
&\sum\limits_{k = 0}^{\infty} \,\left((n + 1 + k)^{(5)} - (n - k)^{(5)}\right) \theta_{n - k}\,\theta_{n + 1 + k} = 0,
\\
&\hspace{4cm}\ldots
\\
&\sum\limits_{k = 0}^{\infty} \,\left((n + 1 + k)^{(2N - 1)} - (n - k)^{(2N - 1)}\right) \theta_{n - k}\,\theta_{n + 1 + k} = 0,
\end{aligned}
\ee
for $n\in\mathbb{Z}$.
\\\\
Due to anticommutativity $\theta_n^2 = 0$. Then expressions \eqref{OddNN} in matrix form are
\begin{scriptsize}
\beq
\renewcommand{\arraystretch}{2.5}
\begin{pmatrix}
(n + 1)^{(1)} - (n - 1)^{(1)} & (n + 2)^{(1)} - (n - 2)^{(1)} & (n + 3)^{(1)} - (n - 3)^{(1)} & \ldots & \\
(n + 1)^{(3)} - (n - 1)^{(3)} & (n + 2)^{(3)} - (n - 2)^{(3)} & (n + 3)^{(3)} - (n - 3)^{(3)} & \ldots & \\
(n + 1)^{(5)} - (n - 1)^{(5)} & (n + 2)^{(5)} - (n - 2)^{(5)} & (n + 3)^{(5)} - (n - 3)^{(5)} & \ldots & \\
\vdots & \vdots & \vdots & \vdots & \\
(n + 1)^{(2N - 1)} - (n - 1)^{(2N - 1)} & (n + 2)^{(2N - 1)} - (n - 2)^{(2N - 1)} & (n + 3)^{(2N - 1)} - (n - 3)^{(2N - 1)} & \ldots & \\
\end{pmatrix}
\begin{pmatrix}
\theta_{n - 1}\,\theta_{n + 1}\\
\theta_{n - 2}\,\theta_{n + 2}\\
\theta_{n - 3}\,\theta_{n + 3}\\
\vdots \\
\vdots 
\end{pmatrix}
= 0.
\ee
\end{scriptsize}
Leading principal $N\times N$ minor of this matrix is nondegenerate. Indeed, $(n + k)^{(2l - 1)} - (n - k)^{(2l - 1)} = 2k^{2l - 1} + O(k^{2l - 3})$ is odd polynomial in $k$. Then by elementary transformations determinant of this minor may be reduced to odd-Vandermonde determinant
\beq
\label{OddVandermonde}
\begin{vmatrix}
x_0 & x_1 & \ldots & x_{N - 1}\\
x_0^3 & x_1^3 & \ldots & x_{N - 1}^3\\
x_0^5 & x_1^5 & \ldots & x_{N - 1}^5\\
\vdots & \vdots & \vdots & \vdots\\
x_0^{2N - 1} & x_1^{2N - 1} & \ldots & x_{N - 1}^{2N - 1}
\end{vmatrix}
\ee
at $x_i = i + 1$.
\\\\
In matrix form expressions \eqref{OddNN+1} are
\begin{scriptsize}
\beq
\renewcommand{\arraystretch}{2.5}
\begin{pmatrix}
(n + 1)^{(1)} - n^{(1)} & (n + 2)^{(1)} - (n - 1)^{(1)} & (n + 3)^{(1)} - (n - 2)^{(1)} & \ldots & \\
(n + 1)^{(3)} - n^{(3)} & (n + 2)^{(3)} - (n - 1)^{(3)} & (n + 3)^{(3)} - (n - 2)^{(3)} & \ldots & \\
(n + 1)^{(5)} - n^{(5)} & (n + 2)^{(5)} - (n - 1)^{(5)} & (n + 3)^{(5)} - (n - 2)^{(5)} & \ldots & \\
\vdots & \vdots & \vdots & \vdots \\
(n + 1)^{(2N - 1)} - n^{(2N - 1)} & (n + 2)^{(2N - 1)} - (n - 1)^{(2N - 1)} & (n + 3)^{(2N - 1)} - (n - 2)^{(2N - 1)} & \ldots & \\
\end{pmatrix}
\begin{pmatrix}
\theta_n\,\theta_{n + 1}\\
\theta_{n - 1}\,\theta_{n + 2}\\
\theta_{n - 2}\,\theta_{n + 3}\\
\vdots \\
\vdots 
\end{pmatrix}
= 0.
\ee
\end{scriptsize}
Leading principal $N\times N$ minor of this matrix is nondegenerate. Indeed, $(n + 1 + k)^{(2l - 1)} - (n - k)^{(2l - 1)} = 2\left(k + \frac{1}{2}\right)^{2l - 1} + O\left(\left(k + \frac{1}{2}\right)^{2l - 3}\right)$ is odd polynomial in $k + \frac{1}{2}$. Then by elementary transformations determinant of this minor may be reduced to odd-Vandermonde determinant \eqref{OddVandermonde} at $x_i = i + \frac{1}{2}$.
\\
Nondegeneracy of leading principal $N\times N$ minor means that any monomial of the form $\theta_{n}\,\theta_{n + k}$ with $0 < |k| \leq 2N$ can be reduced to infinite linear combination of $\theta_i\,\theta_j$ with $|i - j| > 2N$ with relations \eqref{RelationsOdd2N}, which action is correctly defined. Applying argument from the \hyperref[Lemma5.1Proof]{proof} of Lemma \ref{lemmaBasicSubspaceSL} we get
\beq
W_{j}^{\sqrt{2N + 1}} = \mathbb{C}^{F}_{2N}\{\theta_i \:|\: i \leq -j(2N + 1)\}\,|j\sqrt{2N + 1}\rangle.
\ee

\end{proof}

\section{Semi-infinite construction of $V_{\sqrt{N}\mathbb{Z}}$}

\begin{theorem}
~\
    \begin{enumerate}
        \item Even case.
        \\\\
        Let $W^{\sqrt{2N}} = \bigcup\limits_{j\in\mathbb{Z}}\, W^{\sqrt{2N}}_j$. $W^{\sqrt{2N}}$ is invariant under modes of $a(z)$. In particular, $W^{\sqrt{2N}} = V_{\sqrt{2N}\,\mathbb{Z}}$. 
        \item Odd case.
        \\\\
         Let $W^{\sqrt{2N + 1}} = \bigcup\limits_{j\in\mathbb{Z}}\, W^{\sqrt{2N + 1}}_j$. $W^{\sqrt{2N + 1}}$ is invariant under modes of $a(z)$. In particular, ${W^{\sqrt{2N + 1}} = V_{\sqrt{2N + 1}\,\mathbb{Z}}}$.
    \end{enumerate}
\end{theorem}
\begin{proof}
    Proofs of odd and even case are essentially the same. Consider even case
    \beq
    W^{\sqrt{2N}} = \bigcup\limits_{j\in\mathbb{Z}}\, W^{\sqrt{2N}}_j \subset V_{\sqrt{2N}\,\mathbb{Z}}.
    \ee
$W^{\sqrt{2N}}$ is invariant under $e_i$'s by its construction. For any ${j\in\mathbb{Z} \;\; a_{-1}|j\sqrt{2N}\rangle\in W^{\sqrt{2N}}}$. Indeed, consider vector
\beq
e_{-2Nj + N - 1}|(j - 1)\sqrt{2N}\rangle \in W^{\sqrt{2N}}_{j - 1}.
\ee
\beq
\begin{aligned}
&a_1 e_{-2Nj + N - 1}|(j - 1)\sqrt{2N}\rangle = [a_1, e_{-2Nj + N - 1}]|(j - 1)\sqrt{2N}\rangle
\\
& = \sqrt{2N}\,e_{-2Nj + N} |(j - 1)\sqrt{2N}\rangle = \sqrt{2N}\,|j\sqrt{2N}\rangle,
\end{aligned}
\ee
which means 
\beq
e_{-2Nj + N - 1}|(j - 1)\sqrt{2N}\rangle = \sqrt{2N} a_{-1}  |j\sqrt{2N}\rangle. 
\ee
Then with commutation relations $[a_i, e_j] = \sqrt{2N}\,e_{i + j}$ we get invariance of $W^{\sqrt{2N}}$ under $a_i$ with $i \geq -1$.
\\
For any $j\in\mathbb{Z}$\: $a_{-2}|j\sqrt{2N}\rangle~\in~W^{\sqrt{2N}}$. Indeed, consider vector
\beq
e_{-2Nj + N - 2}|(j - 1)\sqrt{2N}\rangle \in W^{\sqrt{2N}}_{j - 1}.
\ee
\beq
\begin{aligned}
&a_2\, e_{-2Nj + N - 2}|(j - 1)\sqrt{2N}\rangle = [a_2, e_{-2Nj + N - 2}]|(j - 1)\sqrt{2N}\rangle 
\\
& = \sqrt{2N}\,e_{-2Nj + N} |(j - 1)\sqrt{2N}\rangle = \sqrt{2N}\,|j\sqrt{2N}\rangle,
\end{aligned}
\ee
which means $e_{-2Nj + N - 2}|(j - 1)\sqrt{2N}\rangle$ has nonzero component along $a_{-2}|j\sqrt{2N}\rangle$. By the previous argument $a_{-1}^2|j\sqrt{2N}\rangle \in W^{\sqrt{2N}}$, then $a_{-2}|j\sqrt{2N}\rangle\in W^{\sqrt{2N}}$.  
\\\\
Then with commutation relations $[a_i, e_j] = \sqrt{2N}\,e_{i + j}$ we get invariance of $W^{\sqrt{2}}$ under $a_i$ with $i \geq -2$.
\\
Acting by induction we get invariance of $W^{\sqrt{2N}}$ under modes of $a(z)$. For that we consider vector
\beq
e_{-2Nj + N - n}|(j - 1)\sqrt{2N}\rangle,
\ee
\beq
\begin{aligned}
&a_n\, e_{-2Nj + N - n}|(j - 1)\sqrt{2N}\rangle = [a_n, e_{-2Nj + N - n}]|(j - 1)\sqrt{2N}\rangle  
\\
&= \sqrt{2N}\,e_{-2Nj + N} |(j - 1)\sqrt{2N}\rangle = \sqrt{2N}\,|j\sqrt{2N}\rangle,
\end{aligned}
\ee
which means $e_{-2Nj + N - n}|(j - 1)\sqrt{2N}\rangle$ has nonzero component along $a_{-n}|j\sqrt{2N}\rangle$. By the induction assumption all vectors $a_{-i_1}\ldots a_{-i_k}|j\sqrt{2N}\rangle$ with $i_j \in \mathbb{N}$, $k > 1$ and $i_1 + i_2 + \ldots + i_k = n$ lie in $W^{\sqrt{2N}}$, consequently $a_{-n}|j\sqrt{2N}\rangle\in W^{\sqrt{2N}}$. Commuting $e_i$'s and $a_{-n}$ we get invariance of $W^{\sqrt{2}}$ under all $a_i$ with $i\geq - n$.

As soon as all vacuum vectors $|j\sqrt{2N}\rangle \in W^{\sqrt{2N}}$ by construction and $W^{\sqrt{2N}}$ is invariant under modes of $a(z)$ , $W^{\sqrt{2N}} = V_{\sqrt{2N}\,\mathbb{Z}}$.
\end{proof}
With natural identification
\beq
|j\sqrt{2N}\rangle \rightarrow e_{(-2j + 1)N}\,e_{(-2j + 3)N}\,e_{(-2j + 5)N}\,\ldots
\ee
it's clear that $V_{\sqrt{2N}\,\mathbb{Z}}$ is the span of semi-infinite monomials
\beq
\label{basis2N}
\left\{\mathbf{e}_a \hspace{2mm}|\hspace{2mm} a\in\Fib^{(N,\, 2N - 1)}_{\infty}\right\}.
\ee
Thus\footnote{Characters of infinite Fibonacci configurations spaces are calculated in \cite{Kenzhaev_Durfee_2023}}
\beq
\ch V_{\sqrt{2N}\,\mathbb{Z}} \leq \ch \Fib^{(N,\, 2N - 1)}_{\infty} = \sum\limits_{m\in\mathbb{Z}}\,\frac{z^m\,q^{Nm^2}}{(q)_{\infty}},
\ee
but we know from formula \eqref{Character2N} that these characters are equal. Then all vectors from \eqref{basis2N} are linearly independent, which proves the even case of Main Theorem \ref{maintheorem}. 
\\\\
With natural identification
\beq
|j\sqrt{2N + 1}\rangle \rightarrow \theta_{(1 - j)(2N + 1)}\,\theta_{(2 - j)(2N + 1)}\,\theta_{(3 - j)(2N + 1)}\,\ldots
\ee
it's clear that $V_{\sqrt{2N + 1}\,\mathbb{Z}}$ is the span of semi-infinite monomials
\beq
\label{basis2N+1}
\left\{\mathbf{\Theta}_a \hspace{2mm}|\hspace{2mm} a\in\Fib^{(0,\, 2N)}_{\infty}\right\},
\ee
Thus,
\beq
\ch V_{\sqrt{2N + 1}\,\mathbb{Z}} \leq \ch \Fib^{(0,\, 2N)}_{\infty} = \sum\limits_{m\in\mathbb{Z}}\,\frac{z^m\,q^{(2N + 1)\frac{m(m - 1)}{2}}}{(q)_{\infty}},
\ee
but we know from formula \eqref{Character2N+1} that these characters are equal. Then all vectors from \eqref{basis2N+1} are linearly independent, which proves the odd case of Main Theorem \ref{maintheorem}.

\begin{consequence}
\label{ConsequenceStructure}
\beq
\mathbb{C}^{F}_{2N - 1}[e_i \:|\: i \leq -2Nj - N] \simeq W_{j}^{\sqrt{2N}},
\ee
where isomorphism is given by $f \rightarrow f\,|j\sqrt{2N}\rangle$, $f\in \mathbb{C}^{F}_{2N - 1}[e_i \:|\: i \leq -2Nj - N]$.
\\
Then character is
\beq
\label{Character2NN}
\ch W_{j}^{\sqrt{2N}} = \sum\limits_{m = j}^{\infty}\,\frac{z^m\,q^{Nm^2}}{(q)_{m - j}}.
\ee
\ben
\mathbb{C}^{F}_{2N}\{\theta_i \:|\: i \leq -j(2N + 1)\} \simeq W_{j}^{\sqrt{2N + 1}},
\een
where isomorphism is given by $f \rightarrow f\,|j\sqrt{2N + 1}\rangle$, $f\in \mathbb{C}^{F}_{2N}\{\theta_i \:|\: i \leq -j(2N + 1)\}$.
\\
Then character is
\beq
\label{Character2NN+1}
\ch W_{j}^{\sqrt{2N}} = \sum\limits_{m = j}^{\infty}\,\frac{z^m\,q^{(2N + 1)\frac{m(m - 1)}{2}}}{(q)_{m - j}}.
\ee
\end{consequence}
\section{Functional realization of the basic subspace}
\label{Frealization}
There are natural surjective homomorphisms
\beq
\label{HomMap}
\begin{aligned}
&\mathbb{C}[e_{-N}, e_{-N - 1}, \ldots]\rightarrow W_{0}^{\sqrt{2N}}, \hspace{5mm} P \rightarrow P\,|0\rangle,
\\
&\mathbb{C}\{\theta_{0}, \theta_{-1}, \ldots\}\rightarrow W_{0}^{\sqrt{2N + 1}}, \hspace{5mm} \tilde{P} \rightarrow \tilde{P}\,|0\rangle.
\end{aligned}
\ee
I.e.
\beq
\label{IdealHomomorphism}
\begin{aligned}
    &W_{0}^{\sqrt{2N}} \simeq \mathbb{C}[e_{-N}, e_{-N - 1}, \ldots]/I_{2N},
    \\
    &W_{0}^{\sqrt{2N + 1}} \simeq \mathbb{C}\{\theta_{0}, \theta_{-1}, \ldots\}/I_{2N + 1},
\end{aligned}
\ee
where $I_{2N}, I_{2N + 1}$ are some two-sided ideals (kernels of maps \eqref{HomMap}). Structure of these ideals is described by the following 
\begin{theorem}
\label{TheoremIdealStructure}
    \begin{enumerate}
		~\\\        
        \item $I_{2N}$ is generated by relations
        \ben
        e_+^2(z) = 0,\quad e_+\cdot e_+''(z),\quad \ldots\,,\quad e_+\cdot e_+^{(2N - 2)}(z) = 0,
        \een
        where $e_+(z) = \sum\limits_{n \leq - N}\, e_n\,z^{-n - N}$.
        \item $I_{2N + 1}$ is generated by relations
        \ben
        \theta_+\cdot \theta_+'(z) = 0, \quad \theta_+\cdot \theta_+'''(z) = 0, \quad \ldots\,,\quad \theta_+\cdot \theta_+^{(2N - 1)}(z) = 0,
        \een
        where $\theta_+(z) = \sum\limits_{n \leq 0}\, \theta_n\,z^{-n}$.
    \end{enumerate}
\end{theorem}
\begin{proof}
    This theorem is a consequence of Lemma \ref{LemmaBasicSubspace2} and Consequence \ref{ConsequenceStructure} , as soon as both sides of \eqref{IdealHomomorphism} are isomorphic to space of Fibonacci polynomials in corresponding variables.
\end{proof}
Space $\mathbb{C}[e_{-N}, e_{-N - 1}, \ldots]$ has a natural pairing with symmetric algebra of space of differential 1-forms on circle of type $z^{N - 1}\,\mathbb{C}[z]\,dz$. Indeed, let
\ben
\mathbb{C}[e_{-N}, e_{-N - 1}, \ldots] = \bigoplus\limits_{m \geq 0}\, V_m,
\een
where $V_m$ is subspace of polynomials of fixed charge (number of $e$'s in monomials) $m$. $V_m$ is paired with $S^m(z^{N - 1}\,\mathbb{C}[z]\,dz)$ by the formula
\ben
\begin{aligned}
&\langle (z_1\,\ldots\,z_m)^{N - 1}\,f(z_1, z_2, \ldots, z_m)\,dz_1\,dz_2\ldots\,dz_m, (\varphi(z_1)\otimes e)\,(\varphi(z_2)\otimes e) \ldots (\varphi(z_m)\otimes e) \rangle
\\
&:= 	\left\{\res_{z_1 =\, \ldots \,= z_m = 0}\,(z_1\,\ldots\,z_m)^{N - 1}\,f(z_1, z_2, \ldots, z_m)\, \varphi(z_1)\,\varphi(z_2)\,\ldots\varphi(z_m)\,dz_1\,dz_2\ldots\,dz_m\right\},
\end{aligned}
\een
where $\varphi_i(x)\in\frac{1}{x^N}\,\mathbb{C}\left[\frac{1}{x}\right]$.
Thus, restricted dual $\left(W_{0}^{\sqrt{2N}}\right)^*$ is naturally isomorphic to annihilator of $I_{2N}$, which is from Theorem \ref{TheoremIdealStructure}
\ben
\left(W_{0}^{\sqrt{2N}}\right)^*_m \simeq \left\{\,\prod\limits_{i < j}\,(z_i - z_j)^{2N}\, f(z_1, z_2, \ldots, z_m)(z_1\,\ldots\,z_m)^{N - 1}\,dz_1\,dz_2\ldots dz_m \:|\: f \text{ is symmetric polynomial } \right\}. 
\een
Then
\ben
\ch W_0^{\sqrt{2N}} = \sum\limits_{m \geq 0}\,\ch \left(W_{0}^{\sqrt{2N}}\right)^*_m = \sum\limits_{m = 0}^{\infty} \frac{z^m\,\left(q^{\frac{m(m - 1)}{2}}\right)^{2N}\,q^{mN}}{(q)_m}= \sum\limits_{m = 0}^{\infty}\,\frac{z^m\,q^{Nm^2}}{(q)_m},
\een
which expectedly coincides with formula \eqref{Character2NN}.
\\
Same argument in odd case gives
\ben
\left(W_{0}^{\sqrt{2N + 1}}\right)^*_m \simeq \left\{\,\prod\limits_{i < j}\,(z_i - z_j)^{2N + 1}\, f(z_1, z_2, \ldots, z_m)\,\frac{dz_1\,dz_2\ldots dz_m}{z_1\,\ldots\, z_m} \:|\: f \text{ is symmetric polynomial } \right\} 
\een
with character
\ben
\ch W_0^{\sqrt{2N + 1}} = \sum\limits_{m \geq 0}\,\ch \left(W_{0}^{\sqrt{2N}}\right)^*_m = \sum\limits_{m = 0}^{\infty}\,\frac{z^m\,q^{(2N + 1)\frac{m(m - 1)}{2}}}{(q)_m},
\een
which expectedly coincides with formula \eqref{Character2NN+1}.

\begin{comment}
\section*{Representational essence of Durfee rectangle identities}
\end{comment}

\bibliography{bibliography}{}

\end{document}